\documentclass{amsart}

\calclayout

\usepackage{amssymb,dsfont}
\usepackage[shortlabels]{enumitem}
\usepackage{mathtools}

\newcommand{\real}{\mathbb{R}}

\newcommand{\prob}{\mathbf{P}}
\newcommand{\ind}{\mathds{1}}

\DeclareMathOperator{\Var}{Var}

\DeclareMathOperator{\ex}{\mathbf{E}}

\newcommand{\abs}[1]{\left\lvert#1\right\rvert}
\newcommand{\set}[1]{\left\{#1\right\}}

\newtheorem{theorem}{Theorem}[section]
\newtheorem{lemma}[theorem]{Lemma}
\newtheorem{proposition}[theorem]{Proposition}
\newtheorem{corollary}[theorem]{Corollary}
\theoremstyle{remark}

\newtheorem{remark}[theorem]{Remark}
\theoremstyle{definition}

\numberwithin{equation}{section}
\allowdisplaybreaks

\begin{document}
\title{Gatheral  double stochastic volatility  model with~Skorokhod  reflection}
\author{Yuliya Mishura}
\address[Y.\,Mishura]{Department of Probability, Statistics and Actuarial Mathematics, Taras Shev\-chen\-ko National University of Kyiv, 64/13, Volodymyrska st., 01601 Kyiv, Ukraine}
\email{yuliyamishura@knu.ua}
\thanks{YM is supported by The Swedish Foundation for Strategic Research, grant UKR24-0004, and by the Japan Science and
Technology Agency CREST, project reference number JPMJCR2115.}
\author{Andrey Pilipenko}
\address[A.\,Pilipenko]{Institute of Mathematics of Ukrainian National Academy of Sciences, Te\-re\-schen\-kiv\-ska st. 3, 01601 Kyiv, Ukraine}
\address[A.\,Pilipenko]{Section of Mathematics, University of Geneva, UNI DUFOUR,
24, rue du Général Dufour,
Case postale 64,
1211 Geneva 4, Switzerland}
\email{pilipenko.ay@gmail.com}
\thanks{AP thanks  the Swiss National Science Foundation for partial support  of the paper (grants No.~IZRIZ0\_226875, No.~200020\_200400, No. 200020\_192129).}
\author{Kostiantyn Ralchenko}
\address[K.\,Ralchenko]{Department of Probability, Statistics and Actuarial Mathematics, Taras Shevchenko National University of Kyiv, 64/13, Volodymyrska st., 01601 Kyiv, Ukraine}
\email{kostiantynralchenko@knu.ua}
\address[K.\,Ralchenko]{School of Technology and Innovations, University of Vaasa, P.O. Box 700, Vaasa, FIN-65101, Finland}
\email{kostiantyn.ralchenko@uwasa.fi}
\thanks{KR is supported by the Research Council of Finland, decision number 367468.}
\thanks{YM and KR acknowledge that the present research is carried out within the frame of the ToppForsk project no.~274410 of the Research Council of Norway with the title STORM: Stochastics for Time-Space Risk Models}
\begin{abstract}
We investigate the Gatheral model of double mean-reverting stochastic volatility, in which the drift term itself follows a mean-reverting process, and the overall model exhibits mean-reverting behavior. We demonstrate that such processes can attain values arbitrarily close to zero and remain near zero for extended periods, making them practically and statistically indistinguishable from zero. To address this issue, we propose a modified model incorporating Skorokhod reflection, which preserves the model's flexibility while preventing volatility from approaching zero.
\end{abstract}
\keywords{Stochastic volatility, Gatheral model, Cox--Ingersoll--Ross process, CKLS model, mean reversion, stochastic differential equation, strong solution, Skorokhod reflection}
\subjclass{60H10, 91G30, 91G80}
\maketitle

\section{Introduction}
The famous Black--Scholes model, being classical and basic in stochastic finance, is nevertheless constantly criticized as not flexible enough and not corresponding to real price changes in the market. In an effort to make this model more flexible and realistic, most researchers introduce into it a stochastic diffusion coefficient, in other words, stochastic volatility. 

As for stochastic volatility models, there are a lot of them and they are very diverse: Hull and White model, Heston model, Constant elasticity of variance model, GARCH model, rough volatility models and others. As a review of some models,  we recommend the paper \cite{DiKuMiYu}.

The choice of a model for stochastic volatility is dictated by completely clear requirements: the model has to be non-negative and not grow too fast, so as not to drive prices too high and also not to decrease to zero, in order not to over- or underestimate the market prices. It should also be flexible enough.

From this point of view, the Gatheral model of double mean reverting   stochastic volatility, in which the drift contains mean reverting process, and the model is mean-reverting itself, is very attractive. That is why we chose this model as the object of our study.

However, models of this type contain some internal danger, which is often neglected: namely, such processes can take on fairly small values  for a fairly long time, that is, be practically and statistically indistinguishable from zero. During this period they lose their flexibility, cannot control the market, and the price is actually subject only to its   drift, which, of course, is unrealistic.

There are various possible approaches to correcting such models. Our approach is based on models with Skorokhod reflection, which simultaneously preserves the flexibility of the model and does not allow volatility to be near zero. The choice of the reflection level is a subject for a separate discussion, we discuss this issue.

The contents of the paper are presented in more detail at the end of the Section \ref{Prelim}, which we will now turn to.  
 
\section{Preliminaries}\label{Prelim}

Let $(\Omega, \mathcal{F}, \prob)$ be a stochastic basis with filtration $\mathbb{F}=(  \mathcal{F}_t)_{t\ge 0}$, and let $(w, W, B)=(w_t, W_t, B_t)_{t\ge 0}$ be three possibly pairwise correlated Wiener processes  with respect to the filtration $\mathbb{F}$. In his seminal paper \cite{gatheral2008}, J.\ Gatheral introduced and convincingly motivated a flexible model for the double mean-reverting dynamics of an asset price of the form
\begin{align}
dS_t &= S_t \sqrt{X_t}\,dw_t,
\label{eq:dS}
\\
dX_t &= (a_1 Y_t - b_1 X_t)\,dt + \sigma_1 X_t^{\alpha_1}\,dW_t,
\label{eq:dX}
\\
dY_t &= (a_2 - b_2 Y_t)\,dt + \sigma_2 Y_t^{\alpha_2}\,dB_t,
\qquad t \ge 0,
\label{eq:dY}
\end{align}
where $a_i\ge 0, b_i>0, \sigma_i > 0$, $\alpha_i \in \left [\frac12, 1\right]$,
$   X_0, Y_0, S_0 > 0.$ We emphasize that throughout the paper all initial values are nonrandom and strictly positive. 

In the case $\alpha_i = \frac12$, $i = 1, 2$, 
the model was called double Heston, 
in the case $\alpha_i = 1$, $i = 1, 2$, 
double lognormal, and in the general case double CEV (Constant Elasticity of Variance) model.
Note that in the case $\alpha_i \in (\frac12, 1)$, a more common name for the respective process is the CKLS (Chan--Karolyi--Longstaff--Sanders) model \cite{ckls1992}.

Since the publication of \cite{gatheral2008}, many papers were devoted to the model \eqref{eq:dS}--\eqref{eq:dY}, including its approximations and respective numerics.
We mention just \cite{albuhayri2022, bayer2013} in this connection, without claiming to present  an exhaustive list.

However, the analytic properties of the model \eqref{eq:dS}--\eqref{eq:dY} have not been sufficiently studied until now.
Some properties (such as existence and uniqueness of solutions and comparison theorems) have been considered as obvious (as they really are such, to some extent), others have not been addressed at all.
What we have in mind: assume that $X$ has ``too many'' zeros, then it is not a suitable model for the asset price \eqref{eq:dS}. Moreover, even if the process $X$ is strictly positive but is very close to zero for some time, the situation is also inappropriate. So, in addition to summarizing the more apparent properties, our paper examines this previously undescribed situation in detail. 

The structure of the next part of the paper is as follows.
Section~\ref{sec-auxiliary} is devoted to the asymptotic properties of the \emph{internal} process $Y$, governed by equation \eqref{eq:dY}. We analyze the behavior of its mean and variance, as well as its pathwise asymptotic behavior. The aim of this section is twofold: to review several known results and to establish new findings of a similar nature that, to the best of our knowledge, have not been previously reported. This section is organized according to the value of the parameter $\alpha_2$: the linear case $\alpha_2 = 1$ is treated in subsection~\ref{sec-auxiliary-1}, the sublinear case $\alpha_2 \in (\frac12, 1)$ (corresponding to the so-called CKLS model) is discussed in subsection~\ref{sec-sublinear}, and the case $\alpha_2 = \frac12$ (corresponding to the Cox--Ingersoll--Ross (CIR) process) is examined in subsection~\ref{sec-CIR}.

Section~\ref{external} investigates the properties of the \emph{external} process $X$. In subsection~\ref{exanduni}, we establish the existence and uniqueness of the solution $(X,Y)$ to the system \eqref{eq:dX}--\eqref{eq:dY}. Subsection~\ref{sec-closetozero} focuses on the behavior of the process $(X,Y)$ near the origin $(0, 0)$, demonstrating, in particular, that it may remain close to this point with non-negligible probability.

To address this behavior, Section~\ref{Reflected-CKLS} introduces a reflected CKLS model, which prevents the trajectories of $Y$ from reaching zero while preserving the mean-reverting property.

The appendix provides supplementary results concerning the existence, uniqueness, and a comparison theorem for the equation \eqref{eq:dX} in a more general setting.

\section{The properties of the internal mean-reverting process depending on the power index of the diffusion coefficient}\label{sec-auxiliary}
In this section we consider the internal  process \eqref{eq:dY}.
Its properties depend on  the value of all coefficients, however, we  will conduct our study depending on the value of the exponent $\alpha_2$. And although all these cases are comparatively well studied, it is still useful to consider in detail the asymptotics of the corresponding processes. We study both   the behavior of their mean and variance, and then consider their trajectory behavior. Sometimes these behaviors  are very different, in the sense that  the behavior of the trajectories does not correspond to the behavior of the numerical characteristics. Note that from the empirical point of view, CKLS processes were carefully systemized in \cite{ckls1992} and then in \cite{borkovec98}, however, in some cases, both trajectory-wise analysis and even the moment's behavior sometimes needs more   calculations.

\subsection{The linear case }\label{sec-auxiliary-1}  
Let $\alpha_2 = 1$. Then properties of the process $Y$ are well known, see for example \cite[Sec.~5.3]{klebaner2012}. 
\begin{proposition}
\begin{itemize}
\item[$(i)$] There exists the unique strong solution of equation \eqref{eq:dY}, and this solution   has a form  
\[
Y_t = \exp\set{R_t} \left(Y_0 + a_2 \int_0^t \exp\set{-R_s}\,ds \right),
\]
where
\[
R_t = -b_2 t + \sigma_2 B_t - \frac{\sigma_2^2}{2}\, t, \quad t \ge 0. 
\] 
\item[$(ii)$] Process  $Y$ is  a.s. strictly positive.
\item[$(iii)$] The process $Y$ is ergodic with a stationary density  that corresponds to inverse gamma distribution and has the following form:
\begin{equation}\label{eq:invgamma}
p_{\infty}(x)=\left(\frac{\sigma_2^2}{2a_2}\right)^{-2b_2/\sigma_2^2-1} \!\left(\Gamma\left(\frac{2b_2}{\sigma_2^2}+1\right)\right)^{-1} \! x^{-2b_2/\sigma_2^2-2} \exp\set{-\frac{\sigma_2^2}{2a_2} x^{-1}},\,  x>0.
\end{equation}
\end{itemize}
\end{proposition}

Now, let us compute the mean and variance of $Y$ and study their asymptotic behavior.
 \begin{lemma}
Let $\alpha_2 = 1$. Then
\begin{gather}
\label{eq:EY}
\ex Y_t  = \left(Y_0 - \frac{a_2}{b_2}\right) e^{-b_2 t} + \frac{a_2}{b_2},
\\
\label{eq:EY^2}
\ex Y_t^2 =  
\begin{cases}
Y_0^2 e^{-(2b_2 - \sigma_2^2) t} + \frac{2a_2}{b_2 - \sigma_2^2} (Y_0 - \frac{a_2}{b_2}) e^{-b_2 t} (1 - e^{-(b_2 - \sigma_2^2) t}) 
\\[2pt]
\qquad {} + \frac{2a_2^2}{b_2 (2b_2 - \sigma_2^2)} \bigl(1 - e^{-(2b_2 - \sigma_2^2) t}\bigr),
& \sigma^2 \ne b_2,\, \sigma^2 \ne 2b_2,
\\[5pt]
Y_0^2 e^{- b_2 t} + 2a_2 (Y_0 - \frac{a_2}{b_2}) t e^{-b_2 t}
+ \frac{2a_2^2}{b_2^2} (1 - e^{-b_2 t}),
& \sigma^2 = b_2,
\\[5pt]
Y_0^2 +
\frac{2a_2}{b_2} (Y_0 - \frac{a_2}{b_2}) (1 - e^{-b_2 t}) + \frac{2a_2^2}{b_2} t,
& \sigma^2 = 2b_2.
\end{cases}
\end{gather}
 \end{lemma}
\begin{proof}
Taking the expectation on both sides of \eqref{eq:dY}, we obtain the following integral equation for $\ex Y_t$:
\begin{equation}\label{eq:EY-equation}
\ex Y_t = Y_0 + \int_0^t (a_2 - b_2 \ex Y_s)\,ds.
\end{equation}
Solving this equation yields \eqref{eq:EY}.

Furthermore, by the It\^o lemma,
\[
Y_t^2 = Y_0^2 + \int_0^t \left(2a_2 Y_s + \left(\sigma_2^2-2b_2\right)Y_s^2\right)ds + 2\sigma_2 \int_0^t Y_s^2 dB_s.
\]
Hence,
\[
\ex Y_t^2 = Y_0^2 + \int_0^t \left(2a_2 \ex Y_s + \left(\sigma_2^2-2b_2\right)\ex Y_s^2\right)ds.
\]
Denoting $f(t) = \ex Y_t^2$ and taking into account \eqref{eq:EY},
we arrive at the following differential equation
\[
f'(t)= ( \sigma_2^2-2b_2) f(t) + 2a_2 \left(Y_0 - \frac{a_2}{b_2}\right) e^{-b_2 t} + \frac{2a_2^2}{b_2},
\qquad f(0) = Y_0^2.
\]
Solving this  differential equation we arrive at \eqref{eq:EY^2}.
\end{proof}

\begin{remark}\label{rem:EY-monot}
Let us study the monotonicity of $\ex Y_t$.
From \eqref{eq:EY}, it follows  that
\begin{itemize}
    \item if $Y_0 > \frac{a_2}{b_2}$, then $\ex Y_t$ decreases from $Y_0$ to $\frac{a_2}{b_2}$ as $t$ increases from 0 to infinity,
    \item if $Y_0 < \frac{a_2}{b_2}$, then $\ex Y_t$ increases from $Y_0$ to $\frac{a_2}{b_2}$ as $t$ increases from 0 to infinity.
\end{itemize}
In particular, we will make use of the fact that $\ex Y_t$ remains bounded away from zero.
\end{remark}

Taking the limit as $t \to \infty$ in \eqref{eq:EY}--\eqref{eq:EY^2}, we obtain the asymptotic values of $\ex Y_t$ and $\ex Y_t^2$, and consequently, of the variance $\Var Y_t$. Note that these asymptotic values can be deduced by computing the moments of stationary distribution \eqref{eq:invgamma}, see \cite[Example~4.3]{borkovec98}.
\begin{corollary}
Let $\alpha_2 = 1$. Then
\[
\ex Y_t \to \frac{a_2}{b_2}, \quad t\to\infty.
\]
Moreover,
\begin{enumerate}[(i)]
\item
if $a_2>0$, then
\[
\Var Y_t \to 
\begin{cases}
\dfrac{a_2^2\sigma_2^2}{b_2^2(2b_2 - \sigma_2^2)}, & \text{if } \sigma_2^2 < 2b_2,\\
\infty,  & \text{if }  \sigma_2^2 \ge 2b_2,
\end{cases}
\quad t\to\infty,
\]
\item
if $a_2=0$, then
\[
\Var Y_t \to 
\begin{cases}
0, & \text{if } \sigma_2^2 < 2b_2,\\
Y_0^2, & \text{if } \sigma_2^2 = 2b_2,\\
\infty,  & \text{if } \sigma_2^2 > 2b_2,
\end{cases}
\quad t\to\infty.
\]
\end{enumerate}
\end{corollary}

\begin{remark}
Of course, the simplest case is when $a_2 = 0$, and $Y_t = Y_0 \exp\set{R_t}$.
In this case for any $b_2 \ge 0$, $R_t \to -\infty$ and $Y_t \to 0$ a.s., as $t \to \infty$, while $\ex Y_t = Y_0 \exp\set{-b_2 t}$ and tends to 0 if $b_2 > 0$ and equals 1 if $b_2 = 0$.
Moreover, variance
$\Var Y_t = Y_0^2 \exp\set{-2 b_2 t}(\exp\{\sigma^2 t\} - 1)$
and can tend to 0 or to $\infty$ depending on whether $2b_2 > \sigma^2$ or $2b_2 < \sigma^2$. In the case $2b_2 = \sigma^2$ $\Var Y_t \to Y_0^2$, $t\to\infty$.
We wish to emphasize here that the case $Y_t \to 0$ a.s.\ while $\Var Y_t \to \infty$ is possible.
\end{remark}

In this connection, let us consider trajectory-wise asymptotic behavior of $Y_t$.

\begin{lemma}
\begin{enumerate}[(i)]
\item Let $a_2 > 0$. Then
\[
\limsup_{t \to \infty} Y_t = +\infty \text{ a.s.\ and }
\liminf_{t \to\infty} Y_t = 0 \text{ a.s.,}
\]
$Y$ is a recurrent process.

\item Let $a_2 = 0$. Then $Y_t \to 0$ a.s.\ as $t \to +\infty$.
\end{enumerate}
\end{lemma}

\begin{proof}
We follow the approach proposed in \cite[Chapter VI, Section 3, p.~446]{ikeda-watanabe}.
Consider the set $I = (0, +\infty)$, i.e., $I = (l,r)$, where $l = 0$, $r = +\infty$.
If $c > 0$, and we consider an SDE
\[
dX_t = b(X_t) dt+\sigma(X_t) dB_t,
\quad X_0 = x_0,\quad t \ge 0,
\]
then create a scale function
\[
s(x) = \int_c^x \exp\set{-\int_c^y \frac{2b(z)}{\sigma^2(z)}\,dz}\,dy.
\]
In our case $b(z) = a_2 - b_2 z$, $\sigma(z) = \sigma_2 z$, therefore
\[
s(x) = c^{-\frac{2b_2}{\sigma_2^2}} \exp\set{-\frac{2a_2}{\sigma_2^2c}}
\int_c^x \exp\set{\frac{2a_2}{\sigma_2^2 y}} y^{\frac{2b_2}{\sigma_2^2}}\,dy.
\]

Now we need to calculate
$s(0)\coloneqq \lim_{x \to 0+} s(x)$
and $s(+\infty)\coloneqq \lim_{x \to +\infty} s(x)$.
Obviously,
\[
\int_c^{0+} \exp\set{\frac{2a_2}{\sigma_2^2 y}} y^{\frac{2b_2}{\sigma_2^2}}\,dy = -\infty,
\]
and so $s(0) = -\infty$, if $a_2 > 0$.
If $a_2 = 0$, then $s(0) = 0$.
Furthermore, $s(+\infty) = +\infty$ for all $a_2 \ge 0$, $b_2 \ge 0$.
According to \cite[Chapter VI, Theorem 3.1]{ikeda-watanabe}, if $s(0) = -\infty$ and $s(+\infty) = +\infty$, then
\[
\prob_{x_0}\left(\limsup_{t\to\infty} Y_t = +\infty\right)
= \prob_{x_0}\left(\liminf_{t\to\infty} Y_t = 0\right) = 1,
\]
for any $x_0 > 0$, and the process $Y$ is recurrent.
Therefore, if $a_2 > 0$, then $Y_t$ is recurrent and oscillates between 0 and $+\infty$.

Furthermore, if $s(0)\in \real$ and $s(+\infty) = +\infty$, then
\[
\prob_{x_0}\left(\lim_{t\to\infty} Y_t = 0\right) = 1,
\]
i.e., $Y_t \to 0$ a.s.\ as $t \to \infty$.
Lemma is proved.
\end{proof}

\subsection{Sublinear case, Chan--Karolyi--Longstaff--Sanders (CKLS) process}
\label{sec-sublinear}
Let\linebreak $\alpha_2 \in (\frac12,1)$.
Then $Y$ is a.s.\ strictly positive, ergodic and has a stationary density. More precisely, 
the existence and uniqueness of a solution from Yamada--Watanabe theorem \cite[Prop.~2.13, p.~ 291]{karatzas-shreve}, the strict positivity follows from Feller's test for explosions \cite[Thm.~5.29, p.~348]{karatzas-shreve}, ergodicity is an application of the ergodic theory for homogeneous diffusions \cite[Ch.~1, \S\,3]{skorokhod89}. The next result was formulated in \cite{MRD22}, see also \cite{AndPit07}.

\begin{proposition}\label{prop:CKLS}
Let $\alpha_2 \in (\frac12,1)$.
\begin{enumerate}
\item \label{st1}
The equation \eqref{eq:dY} has a unique strong solution $Y=\{Y_t,t\ge0\}$.
\item \label{st2}
The process $Y$ is a.s.\ strictly positive.
\item \label{st3}
The process $Y$ is an ergodic diffusion with the following stationary density:
\begin{equation*}
p_{\infty}(x) =  G \cdot x^{-2\alpha_2} \exp\set{\frac{2}{\sigma_2^{2}}\left(\frac{a_2 \cdot x^{1-2\alpha_2}}{1-2\alpha_2} -
\frac{b_2 \cdot x^{2-2\alpha_2}}{2-2\alpha_2} \right)}, \quad x>0.
\end{equation*}
where
\[
G =  \left(\int_{0}^\infty y^{-2\alpha_2} \exp{\frac{2}{\sigma_2^{2}}\left(\frac{a_2 \cdot y^{1-2\alpha_2}}{1-2\alpha_2} - \frac{b_2 \cdot y^{2-2\alpha_2}}{2-2\alpha_2} \right)} dy \right)^{-1}.
\]
\end{enumerate}
\end{proposition}

Now let us study the asymptotic behavior of mean, together with the uniform in $t$  boundedness of  variance (Lemma \ref{lem-2-2}), and trajectory-wise asymptotic behavior (Lemma \ref{lem-2-3}) of CKLS process.

\begin{lemma}\label{lem-2-2}
Let $\alpha_2 \in (\frac12, 1)$.
Then 
\begin{equation}\label{eq:EY-a}
\ex Y_t \to \frac{a_2}{b_2}, \quad \text{as } t \to \infty,
\end{equation}
and there exists $C > 0$ such that
\[
\sup_{t > 0} \ex Y_t^2 \le C.
\]
\end{lemma}

\begin{proof}
Let $\tau_n = \inf\set{t \ge 0 : Y_t \ge n}$, $n \ge 1$.
By the It\^o formula,
\[
Y^2_{t \wedge \tau_n} = Y_0^2 + 2 \int_0^t \left(a_2 Y_{s \wedge \tau_n} - b_2 Y^2_{s \wedge \tau_n}  \right) ds + 2\sigma_2 \int_0^t Y^{1 + \alpha_2}_{s \wedge \tau_n}\, dB_s + \sigma_2^2 \int_0^t Y^{2\alpha_2}_{s \wedge \tau_n}\, ds.
\]
Taking expectations, we get
\[
\ex Y^2_{t \wedge \tau_n} = Y_0^2 + 2 \int_0^t \left(a_2 \ex Y_{s \wedge \tau_n} - b_2 \ex Y^2_{s \wedge \tau_n}  \right) ds + \sigma_2^2 \int_0^t \ex Y^{2\alpha_2}_{s \wedge \tau_n}\, ds,
\]
whence
\begin{align*}
\ex Y^2_{t \wedge \tau_n} 
&\le Y_0^2 + 2 a_2\int_0^t \ex Y_{s \wedge \tau_n} ds + \sigma_2^2 \int_0^t \ex Y^{2\alpha_2}_{s \wedge \tau_n}\, ds
\\
&\le Y_0^2 + a_2 \int_0^t \left(1 + \ex Y^2_{s \wedge \tau_n} \right) ds + \sigma_2^2 \int_0^t \left(1 + \ex Y^2_{s \wedge \tau_n} \right) ds.
\end{align*}
By the Gr\"onwall inequality,
\[
\ex Y^2_{t \wedge \tau_n} \le \left(Y_0^2 + a_2 t + \sigma_2^2 t\right) \exp \set{\left(a_2 + \sigma_2^2\right) t}.
\]
Since it is known that the process $Y$ exists and is unique on any interval, it follows that 
$\tau_n \uparrow \infty$, a.s.\ as $n\to\infty$,
and passing to the limit we get $\ex Y_t^2 < \infty$ for any $t > 0$.
Therefore,
\[
\ex Y_t = Y_0 + \int_0^t \left(a_2 - b_2 \ex Y_s\right) ds,
\]
which coincides with equation \eqref{eq:EY-equation}. Hence, $\ex Y_t$ is given by \eqref{eq:EY}, and \eqref{eq:EY-a} follows. Moreover, $\ex Y_t$ is uniformly bounded in $t$ and bounded away from zero (see Remark~\ref{rem:EY-monot}). Consequently, $\ex Y_t^2$ is also bounded away from zero.
Furthermore,
\begin{equation}\label{eq:EY2}
\ex Y_t^2 + 2b_2 \int_0^t \ex Y_s^2 ds = Y_0^2 + \int_0^t \left(2a_2 \ex Y_s + \sigma_2^2 \ex Y_s^{2\alpha_2} \right) ds.
\end{equation}

Denote
\[
y(t) = \int_0^t \ex Y_s^2\,ds, \quad
R_s \coloneqq 2a_2 \ex Y_s + \sigma_2^2 \ex Y_s^{2\alpha_2}.
\]
Then \eqref{eq:EY2} can be represented in the form of a differential equation:
\[
y'(t) + 2b_2 y(t) = Y_0^2 + \int_0^t R_s\,ds.
\]
Solving it and integrating by parts gives
\[
\ex Y_t^2 = y'(t) = e^{-2b_2 t} Y_0^2 + \int_0^t e^{2b_2(s-t)} R_s\,ds.
\]
In other words, taking into account boundedness of $\ex Y_s$, the fact that $\ex Y_t^2$ is separated from zero, and denoting $C_1$ and $C_2$ constants whose value can change from line to line, we can write 
\begin{align*}
\psi(t) &\coloneqq e^{2 b_2 t} \ex Y_t^2
= Y_0^2 + \int_0^t e^{2 b_2 s} R_s\,ds
= Y_0^2 + \int_0^t e^{2 b_2 s} \left( 2a_2 \ex Y_s + \sigma_2^2 \ex Y_s^{2\alpha_2}\right)\,ds
\\
&\le Y_0^2 + C_1 \int_0^t e^{2 b_2 s}\,ds + C_2 \int_0^t e^{2 b_2 s} \left(\ex Y_s^2\right)^{\alpha_2} ds
\\
&\le C_1 e^{2 b_2 t} + C_2 \int_0^t e^{2 b_2 (1-\alpha_2) s} \left(e^{2 b_2 s} \ex Y_s^2\right)^{\alpha_2} ds,
\end{align*}
whence 
\[
\sup_{0 \le u \le t} \psi(u) \le C_1 e^{2 b_2 t} + C_2 e^{2 b_2 (1-\alpha_2) t} \left(\sup_{0 \le u \le t} \psi(u)\right)^{\alpha_2},
\]
and consequently,
\begin{equation}\label{bumdyk}
\begin{split}
\left(\sup_{0 \le u \le t} \psi(u)\right)^{1-\alpha_2} 
&\le \frac{C_1 e^{2 b_2 t}}{\left(\sup_{0 \le u \le t} \psi(u)\right)^{\alpha_2}} + C_2 e^{2 b_2 (1-\alpha_2) t} 
\\
&\le \frac{C_1 e^{2 b_2(1-\alpha_2) t}}{\left(\ex Y_t^2\right)^{\alpha_2}} + C_2 e^{2 b_2 (1-\alpha_2) t}
\leq C_1e^{2 b_2(1-\alpha_2) t}.
\end{split}
\end{equation}

In particular, we get from \eqref{bumdyk} that    
$\ex Y_t^2 \le C_1$, whence  the proof follows. 
\end{proof}

\begin{lemma}\label{lem-2-3}
Let $\alpha_2 \in (\frac12,1)$.
\begin{enumerate}[(i)]
\item Let $a_2 > 0$.
Then for any initial value $Y_0$
\[
\limsup_{t \to \infty} Y_t = +\infty,\quad
\liminf_{t \to \infty} Y_t = 0\quad \text{ a.s.,}
\]
and $Y$ is a recurrent process.

\item Let $a_2 = 0$. Then $\lim_{t \to \infty} Y_t = 0$ a.s.
\end{enumerate}
\end{lemma}
\begin{proof}
We apply again \cite[Theorem 3.1]{ikeda-watanabe}.
Now, 
\begin{align*}
s(x) &= \int_c^x \exp\set{-2 \int_c^y \frac{a_2 - b_2 z}{\sigma_2^2 z^{2\alpha_2}}\,dz} dy
\\
&= \int_c^x \exp\set{\frac{2 a_2}{\sigma_2^2 \left(2\alpha_2 - 1\right)}
\left(\frac{1}{y^{2\alpha_2 - 1}} - \frac{1}{c^{2\alpha_2 - 1}}\right)
+ \frac{b_2}{\sigma_2^2 \left(1 - \alpha_2\right)}
\left(y^{2 - 2\alpha_2} - c^{2 - 2\alpha_2} \right)} dy
\\
&=\exp\set{-\frac{2 a_2}{\sigma_2^2 \left(2\alpha_2 - 1\right)}
c^{1- 2\alpha_2}
- \frac{b_2}{\sigma_2^2 \left(1 - \alpha_2\right)}
 c^{2 - 2\alpha_2}}
\\*
&\quad\times
\int_c^x \exp\set{\frac{2 a_2}{\sigma_2^2 \left(2\alpha_2 - 1\right)} y^{1 - 2\alpha_2}
+ \frac{b_2}{\sigma_2^2 \left(1 - \alpha_2\right)}
y^{2 - 2\alpha_2}} dy.
\end{align*}
Note that $1 < 2\alpha_2 < 2$.

$(i)$ Therefore in the case $a_2 > 0$
\[
- \lim_{x \to 0+} s(x) = \lim_{x \to +\infty} s(x) = +\infty.
\]
Then the proof follows from item (1) of \cite[Theorem 3.1]{ikeda-watanabe}.

$(ii)$ In the case $a_2 = 0$ 
\[
\lim_{x\to 0+} s(x) > -\infty, \quad
\lim_{x\to +\infty} s(x) = +\infty,
\]
and the proof follows from item (2) of \cite[Theorem 3.1]{ikeda-watanabe}.
\end{proof}
 
\subsection{Cox--Ingersoll--Ross (CIR) process}
\label{sec-CIR}
Let $\alpha_2 = \frac12$.
If $a_2 \ge \frac{\sigma_2^2}{2}$, then $Y$ is a.s.\ strictly positive. If $a_2 < \frac{\sigma^2}{2}$, then $Y$ achieves 0 with probability 1.
In both cases $Y$ is ergodic.
The following statements summarize well-known properties of the CIR process (see, e.g., \cite{AndPit07, borkovec98,dmr22}).
\begin{proposition}\label{prop:CIR}
Let $\alpha_2 = \frac12$.
\begin{enumerate}
\item The equation \eqref{eq:dY} has a unique strong solution $Y=\{Y_t,t\ge0\}$.
\item 
If $a_2\ge \sigma_2^2/2$, then the process $Y$ is a.s.\ strictly positive.
If $0< a_2 < \sigma_2^2/2$, then $Y$ achieves 0 with probability 1, however 0 is a strongly reflecting barrier, in the sense that the   time spent at zero is of Lebesgue measure zero (i.e., the process can touch the barrier, but will leave it immediately).

\item If $a_2 > 0$, then the process $Y$ is ergodic with the following stationary density  that corresponds to gamma distribution:
\[
p_{\infty}(x)=\left(\frac{2b_2}{\sigma_2^2}\right)^{\frac{2a_2}{\sigma_2^2}}x^{\frac{2a_2}{\sigma_2^2}-1} \exp\set{-\frac{2b_2}{\sigma_2^2} x} \Bigm/ \Gamma\left(\frac{2a_2}{\sigma_2^2}\right), \quad x>0.
\]
\end{enumerate}
\end{proposition}

\begin{remark}
Note that when $b_2 = 0$, the process $Y$ reduces to a squared Bessel process, which is non-ergodic. For a detailed discussion on the squared Bessel process and its comparison with the CIR model, we refer the reader to the recent study \cite{MRK25} and the references therein.
\end{remark}

\begin{lemma}\label{l:CIR-moments}
Let $\alpha_2 = \frac12$.
The first two moments of $Y_t$ are equal to
\begin{gather*}
\ex Y_t  = \left(Y_0 - \frac{a_2}{b_2}\right) e^{-b_2 t} + \frac{a_2}{b_2},
\\
\ex Y_t^2 = Y_0^2 e^{-2b_2t} + \frac{Y_0 (\sigma_2^2 + 2 a_2)}{b_2} \left(e^{-b_2t} - e^{-2b_2t}\right)
+ \frac{a_2(\sigma_2^2 + 2a_2)}{2b_2^2}\left(1 - e^{-b_2 t}\right)^2.
\end{gather*}
Hence,
\[
\ex Y_t \to \frac{a_2}{b_2},
\qquad
\Var Y_t \to \frac{a_2 \sigma_2^2}{2 b_2^2}
\qquad\text{as } t \to \infty.
\]
\end{lemma}

\begin{lemma}\label{l:CIR-paths}
Let $\alpha_2 = \frac12$.
\begin{enumerate}[(i)]
\item Let $a_2 \ge \frac{\sigma_2^2}{2}$.
Then for any initial value $Y_0$
\[
\limsup_{t \to\infty} Y_t = +\infty,\quad
\liminf_{t \to \infty} Y_t = 0\quad \text{ a.s.,}
\]
and $Y$ is a recurrent process.

\item Let $0 \le a_2 < \frac{\sigma_2^2}{2}$. Then $\lim_{t \to \infty} Y_t = 0$ a.s.
\end{enumerate}
\end{lemma}

\begin{proof}
We apply \cite[Theorem 3.1, Chapter VI]{ikeda-watanabe}. First, we compute
\begin{align*}
s(x) &= \int_c^x \exp\set{-2 \int_c^y \frac{a_2 - b_2 z}{\sigma_2^2 z}\,dz} dy
= \int_c^x \left(\frac{y}{c}\right)^{-\frac{2a_2}{\sigma_2^2}}\exp\set{\frac{2b_2}{\sigma_2^2}(y-c)}\, dy
\\
&= c^{\frac{2a_2}{\sigma_2^2}} \exp\set{-\frac{2b_2 c}{\sigma_2^2}} \int_c^x y^{-\frac{2a_2}{\sigma_2^2}}\exp\set{\frac{2b_2}{\sigma_2^2} y}\, dy.
\end{align*}
Since the integrand
$y^{-\frac{2a_2}{\sigma_2^2}}\exp\set{\frac{2b_2}{\sigma_2^2} y}$
is positive and tends to infinity as $y \uparrow \infty$, it follows that $s(x) \to +\infty$ as $x \uparrow \infty$ for any $a_2 \geq 0$.

Moreover, 
\[
\lim_{x \to 0} s(x) = - c^{\frac{2a_2}{\sigma_2^2}} \exp\set{-\frac{2b_2 c}{\sigma_2^2}} \int_0^c y^{-\frac{2a_2}{\sigma_2^2}}\exp\set{\frac{2b_2}{\sigma_2^2} y}\, dy
\]
is either $-\infty$ or finite, depending on whether $\frac{2a_2}{\sigma_2^2} \geq 1$ or $\frac{2a_2}{\sigma_2^2} < 1$.

Thus, the conclusions of (i) and (ii) follow from statements (1) and (2) of \cite[Theorem 3.1]{ikeda-watanabe}, respectively.
\end{proof}

\section{Properties of the  external  process $X$}\label{external}

Now we turn to the properties of   the  external  process   $X$, being interested   in the impact of the  internal  process on the properties of the external one. 

\subsection{Existence and uniqueness results}\label{exanduni}

First, we consider the system of equations \eqref{eq:dX}--\eqref{eq:dY} and establish existence-uniqueness result. 
\begin{theorem}\label{theorsyst}
The system of equations \eqref{eq:dX}--\eqref{eq:dY} has the unique strong solution, both processes $X$ and $Y$ are non-negative, and the solution is a strong Markov process.
\end{theorem}

 \begin{proof}
System  \eqref{eq:dX}--\eqref{eq:dY} can be considered as two-dimensional diffusion equation with linear drift and diffusion coefficient that consists of power functions of the form $\sigma_i (x^+)^{\alpha_i}$, $i = 1, 2$, with power indices $\alpha_i$ not exceeding 1.  It means  that all coefficients are continuous functions of at most linear growth. Then it follows from the existence theorem  proved in \cite[p.~59]{skor61},  that this system has an $\mathbb{F}$-adapted solution $(\overline X, \overline Y)$. Moreover, this solution is obviously unique for the equation \eqref{eq:dY}, and being  non-explosive, is a continuous stochastic process on any interval. Then, establishing uniqueness of the solution of equation \eqref{eq:dX}, we, as usual,  consider two solutions and subtract them. Since  process $Y$ disappears after subtraction, the uniqueness can be  proved by the same steps as in Yamada theorem, see  e.g., \cite[Theorem 3.2, p.~182]{ikeda-watanabe}.

Recall that $Y$ and $a_1$ are non-negative, therefore, $\overline Y = Y$.
Now we shall prove the non-negativity of $\overline X$, which will allow us to transit from the diffusion coefficient $\sigma_1 (x^+)^{\alpha_1}$ to $\sigma_1 x^{\alpha_1}$ and identify $\overline X$ with $X$.  
Application of the comparison theorem (Theorem \ref{comp} in Appendix) shows that $\overline X_t$ exceeds a solution to the equation
\[
d\widehat X_t  =- b_1 \widehat X_t\,dt + \sigma_1 |\widehat X_t|^{\alpha_1}\,dW_t,\quad \widehat X_0=0,
\]
 which is identically equal to 0 because of uniqueness and existence of a strong solution.
 That is, $\overline X = X$ is a.s.\ non-negative process satisfying \eqref{eq:dX}.

The strong Markov property follows from existence and uniqueness of the solution.
Theorem is proved.  
 \end{proof}

\begin{remark}
It is possible to prove a more general result about existence and uniqueness solution of equation \eqref{eq:dX}, without assuming that $Y$ is a solution to \eqref{eq:dY}. This is addressed in Theorem \ref{th:ex-un} in the Appendix.
\end{remark}

\subsection{Volatility process may remain quite close to zero for some time, with probability far from zero}\label{sec-closetozero}

Our next objective, which is one of the main objectives of the paper,  is to investigate the behavior of the vector process $(X,Y)$ in a neighborhood of the point $(0,0)$. More precisely, we establish the following two facts: in Theorem \ref{thm:access} it is proved that for any $\alpha_1,\alpha_2\in[\frac12,1)$ and any initial condition, the process $(X, Y)$ enters every open square of the form $(0, \epsilon)^2$ with probability one.
In Theorem~\ref{thm:hitting}, we consider specific cases where one of the parameters, either $\alpha_1$ or $\alpha_2$, is equal to $\tfrac{1}{2}$. In these settings, we can establish a stronger result: namely, if $\alpha_1 = \tfrac{1}{2}$ (resp.\ $\alpha_2 = \tfrac{1}{2}$), then with probability one, the process $X$ (resp.\ $Y$) hits zero while the other process $Y$ (resp.\ $X$) becomes arbitrarily small.
We now state our main results, Theorems~\ref{thm:access} and~\ref{thm:hitting}, followed by the auxiliary Lemmas~\ref{lem:hit1}--\ref{lem:enterR3}, and then provide the proofs of the main theorems.

\begin{theorem}\label{thm:access}
Assume that $\alpha_1,\alpha_2\in[\frac12,1)$. Then the process $(X,Y)$ is recurrent, i.e., for any  nonempty
open set $G\subset (0,\infty)^2$ and any initial starting point
\[
\prob(\forall t_0\, \exists t\geq t_0 \colon (X_t,Y_t)\in G) = 1.
\]
\end{theorem}

\begin{theorem}\label{thm:hitting}
\leavevmode
\begin{enumerate}[1)] 
\item Assume that $\alpha_1=\frac12$, $\alpha_2\in[\frac12,1)$. Then for any $\epsilon>0$ and any starting point
\[
\prob(\forall t_0\, \exists t\geq t_0\colon  X_t=0, Y_t\in[0,\epsilon])=1.
\]
\item Assume that $\alpha_2=\frac12$, $\alpha_1\in[\frac12,1)$. Then for any $\epsilon>0$ and any starting point
\[
\prob(\forall t_0\, \exists t\geq t_0\colon  Y_t=0, X_t\in[0,\epsilon])=1.
\]
\end{enumerate}
\end{theorem}

\begin{remark}
    We don't know if the process hits the origin with positive probability if $\alpha_1=\alpha_2=\frac{1}{2}.$
\end{remark}

\begin{remark}
If $\alpha_1\in(\frac12,1]$, then $\prob(\exists t>0 \colon X_t=0)=0$, and similarly, if $\alpha_2\in(\frac12,1]$, then $\prob(\exists t>0 \colon Y_t=0)=0$. 
The statement for $Y$ follows directly from Feller's test for explosions \cite[Theorem 5.29, p.\ 348]{karatzas-shreve}.
The result for $X$ is obtained by comparing it with the solution of an equation with $a_1 = 0$, using Theorem~\ref{th:comparison}, followed by an application of Feller's test.
\end{remark}

\begin{lemma}\label{lem:hit1}
    Assume that there exists a compact set $K\subset (0,\infty)^2$ such that $(X,Y)$ visits $K$ with probability 1 for any starting point:
\begin{equation}
    \forall (x_1,y_1)\in [ 0,\infty)^2\quad \prob_{x_1,y_1}(\exists t\geq 0\colon (X_t, Y_t)\in K)=1.
    \label{eq:hittingK}
\end{equation}
Then the statement of the Theorem \ref{thm:access} holds true.
\end{lemma} 
\begin{proof}
Let  $F\subset (0,\infty)^2$ be a compact set with smooth boundary such that $K, G^0\subset F^0$, where $A^0$ is the interior of a set $A$. 
It is well known that   the transition probability density function $p_t \bigl((x_1,y_1),(x_2,y_2)\bigr) $ of  $(X_t, Y_t)$  killed at the boundary $\partial F$   is a fundamental solution to the Dirichlet problem
\[
\begin{cases}
\partial_t u = \mathcal A u,
\\
u|_{\partial G}  = 0,  
\end{cases}
\]
where 
\begin{equation}\label{eq:generator}
\mathcal A = (a_1 y - b_1 x) \frac{\partial}{\partial x} + (a_2 - b_2 y)\frac{\partial}{\partial y} + \frac12 \sigma_1^2 x^{2\alpha_1} \frac{\partial^2}{\partial x^2} + \frac12 \sigma_2^2 y^{2 \alpha_2} \frac{\partial^2}{\partial y^2}
\end{equation}
is the generator of the diffusion process $(X,Y)$.
Since the coefficients of the equation are 
infinitely differentiable and satisfy the ellipticity property in $K$, for any fixed $t>0$ the function $p_t \bigl((x_1,y_1),(x_2,y_2)\bigr) $ is continuous in $(x_1, y_1)$, $(x_2, y_2)$ and strictly positive whenever both arguments lie in $F^0$.
Hence, for any fixed $t>0$
\[
\inf_{(x_1,y_1)\in K,\, (x_2,y_2)\in G }p_t \bigl((x_1,y_1),(x_2,y_2)\bigr)>0
\]
and consequently
 \begin{equation}\label{eq:prob_separ}
\alpha(t)\coloneqq\inf_{(x_1,y_1)\in K}\prob_{x_1,y_1} \bigl((X_t,Y_t)\in G\bigr)>0.     
 \end{equation}
Fix an arbitrary $t^*>0$ and introduce stopping times
\[
\sigma_0 \coloneqq 0,\quad \sigma_{n+1} \coloneqq \inf\set{s\geq \sigma_n+t^*: (X_s,Y_s)\in K}.
\]
It follows from \eqref{eq:hittingK} that $\sigma_n<\infty$ a.s.\ for all $n\geq 1$. The strong Markov property of $(X,Y)$ and \eqref{eq:prob_separ} imply that
\[
\prob \bigl(\exists s\in [0,\sigma_n]\colon (X_s,Y_s)\in G\bigr)
\geq 1 - \bigl(1 - \alpha(t_*)\bigr)^n \to 1, \quad n \to \infty. \qedhere
\]
\end{proof}

The following lemma shows the process almost surely enters a sufficiently large compact.
\begin{lemma}\label{l:R0}
There exists a constant $R_0 > 0$ such that for all initial conditions $x$ and $y$,
\begin{equation}\label{eq:hitsline}
\prob_{x,y} \left(\exists t \ge 0\colon X_t + Y_t = R_0 \right) = 1.
\end{equation}
\end{lemma}

\begin{proof}
If $x + y \le R_0$ for some parameter $R_0>0$, then equality \eqref{eq:hitsline} is obvious because $\prob(\limsup_{t\to\infty} Y_t=+\infty)=1$. Hence, it suffices to show existence of $R_0$ such that \eqref{eq:hitsline} 
is satisfied for all $x,y\geq 0$ such that $x + y > R_0$.

The proof employs the Lyapunov function method.
Fix $k>0$ and define the Lyapunov function
\[
V(x,y) = y^2 + k^2 x^2.
\]
Let as before $\mathcal A$ denote the infinitesimal generator of the two-dimensional SDE system \eqref{eq:dX}--\eqref{eq:dY}, given by \eqref{eq:generator}. Applying 
$\mathcal A$ to $V$, we get:
\begin{align*}
\mathcal A V(x,y) &= 2 y (a_2 - b_2 y) + \sigma_2^2 y^{2\alpha_2} + 2 k^2 x (a_1 y - b_1 x) + \sigma_1^2 k^2 x^{2\alpha_1}
\\
&= - 2 b_2 y^2 - 2 b_1 k^2 x^2 + 2 a_1 k^2 x y + o (V(x,y)), \quad |x|+|y|\to\infty.
\end{align*}
Using the elementary inequality 
$2\alpha\beta \le \alpha^2 + \beta^2$, and choosing
$k>0$ sufficiently small, we obtain
\begin{align*}
2 b_2 y^2 + 2 b_1 k^2 x^2 - 2 a_1 k^2 x y
&\ge 2 b_2 y^2 + 2 b_1 (kx)^2 - k a_1 \left((kx)^2 + y^2\right)
\\
&\ge K_1 \left((kx)^2 + y^2\right)
= K_1 V(x,y).
\end{align*}
for some $K_1 > 0$.

Therefore, there exist constants $K_1, K_2, K_3 > 0$ and $R_0 > 0$ such that for all
$x,y \ge 0$, $x+y\ge R_0$ 
\[
\mathcal A V(x,y) \le K_2 - K_1 V(x,y) \le - K_3.
\]

Define the stopping time
\[
\tau_{R_0} \coloneqq \inf\set{t \ge 0: X_t + Y_t = R_0}.
\]
By the It\^o formula, we have for any $x,y \ge 0$, $x+y\ge R_0$,
\[
0 \le \ex V\left(X_{\tau_{R_0} \wedge t}, Y_{\tau_{R_0} \wedge t}\right)
= V(x,y) + \int_0^{\tau_{R_0} \wedge t} \ex \mathcal A V(X_s,Y_s)\,ds
\le V(x,y) - K_3 \ex (\tau_{R_0} \wedge t).
\]
Letting $t\to\infty$, we get 
\[
0 \le V(x,y) - K_3 \ex \tau_{R_0},
\]
which implies $\ex \tau_{R_0} < \infty$.
Hence,
\[
\tau_{R_0} < \infty \quad \text{a.s.} \qedhere
\]
\end{proof}

Note that the set 
$\{(x,y)\in\mathbb{R} : x, y\geq0,\, x+y=R_0\}\not\subset (0,+\infty)^2$,
and therefore Lemma~\ref{lem:hit1} cannot yet be applied to this set to establish Theorem~\ref{thm:access}.
However, in the next statement, we identify a compact set that the process visits with probability greater than $1/2$.
\begin{lemma}\label{lem:enterR3}
    There is $\delta > 0$ such that 
    \[
     \inf_{x,y\geq 0, \, x+y=R_0}\prob_{x,y} \left(\exists t \ge 0\colon |X_t+Y_t-R_0|\leq \frac{R_0}{3},\,  X_t\geq \delta,\, Y_t\geq \delta\right) \geq \frac{1}{2}.
    \]
    Here $R_0$ is from Lemma~\ref{l:R0}.
\end{lemma}
\begin{proof}
We will select sufficiently small $\delta$ such that $\delta\in(0,\frac{R_0}{3})$. Then to prove the Lemma it suffices to show that
\[
     \inf_{ x+y=R_0,\, y\in[0,\delta]}\prob_{x,y} \left(\exists t \ge 0\colon |X_t+Y_t-R_0|\leq \frac{R_0}{3},\,   Y_t\geq \delta\right) \geq \frac{1}{2},
    \]
and
\[
\inf_{ x+y=R_0,\, x\in[0,\delta]}\prob_{x,y} \left(\exists t \ge 0\colon |X_t+Y_t-R_0|\leq \frac{R_0}{3},\,  X_t\geq \delta \right) \geq \frac{1}{2},
\]
Since coefficients of the system \eqref{eq:dX}--\eqref{eq:dY} are continuous, they are bounded on compact sets. Therefore, there exists a sufficiently small $t_0$ such that
\begin{align*}
\MoveEqLeft
\sup_{x,y\geq 0, \, x+y=R_0}\prob_{x,y}\left( \sup_{t\in[0,t_0]} |X_t+Y_t-R_0|> \frac{R_0}{3} \right)
\\*
&= \sup_{x,y\geq 0, \, x+y=R_0}\prob_{x,y}\left( \sup_{t\in[0,t_0]} |X_t + Y_t - (X_0 + Y_0)|> \frac{R_0}{3} \right)
\\
&\leq  \sup_{x,y\geq 0, \, x+y=R_0}\prob_{x,y}\left( \sup_{t\in[0,t_0]} |X_t - X_0|> \frac{R_0}{6},\; \sup_{t\in[0,t_0]} |Y_t -Y_0|> \frac{R_0}{6} \right) < \frac{1}{4}.
\end{align*}
Set $\tau\coloneq\inf\{t\geq 0 : \sup_{t\in[0,t_0]} |X_t - X_0|> \frac{R_0}{6},\; \sup_{t\in[0,t_0]} |Y_t - Y_0|> \frac{R_0}{6}\}$.
By the comparison theorem  
\[ 
\prob(Y_t\geq \bar Y_t,\, t\geq 0)=1,
\]
where $\bar Y$ is the solution of \eqref{eq:dY}, started from 0. Hence
\begin{align*}
\MoveEqLeft[3]\inf_{ x+y=R_0,\, y\in[0,\delta]}\prob_{x,y} \left(\exists t \ge 0\colon |X_t+Y_t-R_0|\leq \frac{R_0}{3},\,   Y_t\geq \delta\right)\\
    & \geq  \inf_{ x+y=R_0,\, y\in[0,\delta]}\prob_{x,y} \left(\exists t \in[0,\tau]\colon  Y_t\geq \delta\right) 
    \\
    &\geq 
    \inf_{x,y\geq 0,\, x+y=R_0 }\prob_{x,y} \left(\{\tau\ge t_0\}\cap\{\exists t\in[0,t_0]\colon   \bar Y_t\geq \delta\}\right)
     \\
     &\geq \frac34 -\prob\left(  \exists t\in[0,t_0]\colon  \bar Y_t\geq \delta\right).
\end{align*}
For any $t_0>0$ 
\[
\prob(\bar Y(t)=0,\, t\in[0,t_0])=1.
\]
Hence, for sufficiently small $\delta>0$
\[
\prob\left(\forall t \in[0,t_0]\colon  \bar Y_t< \delta\right)< \frac{1}{4},
\]
and, consequently, 
\[
\inf_{ x+y=R_0,\, y\in[0,\delta]}\prob_{x,y} \left(\exists t \ge 0\colon |X_t+Y_t-R_0|\leq \frac{R_0}{3},\,   Y_t\geq \delta\right)\ge \frac34-\frac14=\frac12.
\]

Note that for any starting point $(x,y)$,  $x+y=R_0$,
\[
\prob(\forall t\in[0,\tau]\colon   X_t\geq  \bar X_t)=1,
\]
$\bar X$ is a solution to the equation
\[
d\bar X_t  = \left(a_1 \frac{5R_0}{6} - b_1 \bar X_t\right) dt + \sigma_1 \bar X_t^{\alpha_1}\,dW_t,
\]
 with the initial condition $\bar X_0 = 0$. Similarly to the previous calculations,
\begin{align*}
\MoveEqLeft[3]\inf_{ x+y=R_0,\, x\in[0,\delta]}\prob_{x,y} \left(\exists t \ge 0\colon |X_t+Y_t-R_0|\leq \frac{R_0}{3},\,   X_t\geq \delta\right)
\\
&\geq  \inf_{ x+y=R_0,\, y\in[0,\delta]}\prob_{x,y} \left(\exists t \in[0,\tau]\colon  X_t\geq \delta\right)-\frac14
\\
&\geq  \inf_{x,y\geq 0,\, x+y=R_0 }\prob_{x,y} \left(\exists t \in[0,\tau]\colon  \bar X_t\geq \delta\right)-\frac14\geq \frac12,
\end{align*} 
for sufficiently small $\delta>0.$
  This concludes the proof. 
\end{proof}

\begin{proof}[Proof of Theorem \ref{thm:access}]
Since
 \begin{align*}
 \MoveEqLeft
    \prob(\forall t_0\, \exists t\geq t_0 \colon (X_t,Y_t)\in G)=
\lim_{t_0\to\infty}\prob(\exists t\geq t_0 \colon (X_t,Y_t)\in G)\\
&=\lim_{t_0\to\infty}\int_0^\infty\!\!\int_0^\infty\prob_{x,y}( \exists s\ge 0 \colon (X_s, Y_s)\in G)\prob_{X_{t_0},Y_{t_0}}(dx,dy),
 \end{align*}
 to prove the theorem it suffices to verify that
 \[
 \forall x,y\ge 0\colon\quad \prob_{x,y}( \exists t\ge 0\colon (X_t, Y_t)\in G)=1.
 \]
  In order to check this condition, we will   apply Lemma \ref{lem:hit1} with  
$K\coloneqq\{(x,y)\colon x,y\geq \delta$, $|x+y-R_0|\leq \frac{R_0}{3}\} $. 
 Here $R_0, \delta$ are from Lemma \ref{lem:enterR3}.
Introduce stopping times $\tau_n$, $\sigma_n$ as follows:
$\tau_0 = \sigma_0 \coloneqq 0$,
\begin{gather*}
\tau_{n+1}\coloneqq\inf\{t\geq \sigma_n:  X_t+Y_t=R_0 \},
\\
\sigma_{n+1}\coloneqq\inf\{t\geq \tau_{n+1}+t_0: Y_t=2R_0\},
\end{gather*}
where $t_0$ is from Lemma \ref{lem:enterR3}.

Recall that 
\[
\prob\left(\limsup_{t\to\infty} Y_t = +\infty\right)=1,
\]
whence  $\tau_n, \sigma_n $ are finite a.s.
We will say that ``success'' occurs, if for some $t \in [\sigma_k,\sigma_{k+1}]$ the process $(X,Y)$ hits the compact set $K$. As we have shown in Lemma \ref{lem:enterR3}, the probability of ``success'' is not less than $\frac12$. Hence, the probability of not hitting $K$ even once on the time interval $[0,\sigma_n]$ does not exceed $\frac{1}{2^n}$. Thus
\[
\prob(\exists t \in [0, \sigma_n] \colon  (X_t,Y_t)\in K)\geq 1-\frac{1}{2^n},
\]
therefore 
\[
\prob(\exists t \geq 0 \colon  (X_t,Y_t)\in K)=1.
\]
The application of Lemma \ref{lem:hit1} completes the proof.
\end{proof}
\begin{proof}[Proof of Theorem \ref{thm:hitting}]
We will prove only the first part of the theorem; the proof of the second part is similar (but easier).  Assume that $\alpha_1=\frac12$, $\alpha_2\in[\frac12,1)$, and $\epsilon\in(0,1)$. 

Since coefficients of \eqref{eq:dX}--\eqref{eq:dY} are continuous, we can find $t_1 > 0$ small enough such that
    \[
     \sup_{x,y\in[0,2]}\prob_{x,y} \left( \sup_{t\in[0,t_1]}(|X_t-x|+| Y_t-y|)> \frac{\epsilon}{2} \right) < \frac{1}{4}.
    \]
Let us introduce a quadratic Bessel process $\bar X^{(x)}$ satisfying  the equation
\begin{equation}\label{eq:bessel}
d\bar X^{(x)}_t  =   a_1\epsilon \,dt + \sigma_1 \sqrt{  \bar X^{(x)}_t }\,dW_t
 \end{equation}
with initial condition $\bar X^{(x)}_0=x$.
Let us compare this process with the process $X$ (which satisfies equation  \eqref{eq:dX} with $\alpha_1 = \frac12$ and with initial condition $X_0 = x$). Note that the equations \eqref{eq:dX} and \eqref{eq:bessel} have the same diffusion coefficients and their drift coefficients for $y\in[0,\epsilon]$ can be compared as follows:
\[
a_1 y - b_1 x \le a_1 y \le a_1 \epsilon,
\quad x\geq 0,\; y\in[0,\epsilon].
\]
Hence, it follows from comparison principle that
\[
\forall x\geq 0,\; y\in[0,\tfrac{\epsilon}{2}]   \quad \prob_{x,y}\left(\forall t\in[0,\tau]\colon X_t\leq  \bar X^{(x)}_t \right) = 1.
    \]
     where
$\tau \coloneqq \inf\{t\geq 0: | Y_t-y|\geq  \frac{\epsilon}{2}\}$.

It is well known \cite[Ch.\ XI, \S1]{RevuzYor} that if $\epsilon$ is small enough, then the quadratic Bessel process $\bar X^{(x)}$ hits 0 with probability 1. Moreover, 
$\sigma^{(x)}_0\overset{\prob}{\to}0$ as $x\to0+$, where $\sigma^{(x)}_0=\inf\{t\geq 0:  \bar X^{(x)}_t=0 \}$.  Hence, there exist $t_1 > 0$ and $\delta>0$ such that
\[
\sup_{x\in[0,\delta]}\prob\left(\sigma^{(x)}_0\geq t_1\right)\leq \frac{1}{4}.
\]
Therefore, for every $x\in[0,\delta]$, $y\in [0,\frac{\epsilon}{2}]$:
\[
\prob\left(\left\{\sup_{t\in[0,t_1]}Y_t\leq \epsilon\right\}\cap\left\{\exists t\in[0,t_1]\colon X_t=0\right\}\right)\geq \frac{1}{2}.
\]
It follows from Theorem \ref{thm:access} that $(X,Y)$ visits the set $(0, \delta) \times (0, \frac{\epsilon}{2})$ with probability 1 for any initial distribution. The rest of the proof is similar to the proof of Theorem   \ref{thm:access}.
\end{proof}
  
\section{Reflected CKLS model}\label{Reflected-CKLS}
Since we are interested in the number of zeros of $X$ that, in turn, depend on the values of $Y$, it is nice to fix the behavior of trajectories of $Y$ separating them from zero. 
So let us consider the process \cite{chang2015,shi2025} with one-sided reflection (sometimes, two-sided reflection is considered \cite{chang2015, shi2025}, but for us now just the situation with   zeros and how to avoid for $X$ to be near zero, is important, therefore, we separate $Y$ from some lower positive level  and do not introduce upper reflection).
So, for $\alpha_2 \in [\frac12,1)$ we consider the reflected process  $Y^{(m)}$ such that $Y^{(m)}_t\ge m$ and $Y^{(m)}$ satisfies the following stochastic differential equation
\begin{equation}\label{eq:Y-reflected}
dY^{(m)}_t = \left (a_2 - b_2 Y^{(m)}_t\right) dt + \sigma_2 \left(Y^{(m)}_t\right)^{\alpha_2} dB_t + dL_t^{(m)},
\end{equation}
where $0 < m < Y_0$,
$L^{(m)} = \{L^{(m)}_t, t \ge 0\}$ is adapted to the filtration generated by the Wiener process $B = \{B_t, t \ge 0\}$,
the trajectories of $L^{(m)}$ are a.s.\ continuous and nondecreasing, $L^{(m)}_0 = 0$, $L^{(m)}$ increases only on the set $A = \{t : Y^{(m)}_t = m\}$, so that
\[
\int_0^t \ind_A  (Y^{(m)}_s) dL^{(m)}_s = L^{(m)}_t, \quad t \ge 0.
\]
In this case $L^{(m)}$ is the smallest nondecreasing function for which the process \eqref{eq:Y-reflected} is above or equals $m$,
and, according to the well-known solution of the Skorokhod reflection problem (see, e.g., \cite{pilipenko2014}),
\[
L^{(m)}_t = \sup_{0 \le s \le t} \left(\left(m-Z^{(m)}_s\right) \vee 0\right),
\]
where $Z^{(m)}_t= Y_0+\int_0^t(a_2 - b_2 Y^{(m)}_s) ds + \sigma_2\int_0^t (Y^{(m)}_s)^{\alpha_2} dB_s$, and $\{Y^{(m)}_t, t\ge 0\}$ is the unique solution of the equation \eqref{eq:Y-reflected}.

The next result is contained in \cite{pilipenko2014}.
\begin{lemma}\label{moments}  For any $T>0$,     $m>0$, $p>0$ and $\alpha_2\in[1/2,1]$
$$\ex \sup_{0\le t\le T}(Y^{(m)}_t)^p< \infty.$$
    
\end{lemma}

\begin{lemma}\label{lemmeanrev}
The process $Y^{(m)}$ is mean-reverting with the same constant as $Y^{(0)}$; namely,
\[
\lim_{t\to\infty}\ex Y^{(m)}_t = \frac{a_2}{b_2}.
\]
\end{lemma}

\begin{proof}
Equation \eqref{eq:Y-reflected} can be rewritten as
\begin{align*}
Y^{(m)}_t - L^{(m)}_t &= Y_0 + a_2 t
- b_2 \int_0^t \left(Y^{(m)}_s - L^{(m)}_s\right) ds
\\*
&\quad+ \sigma_2 \int_0^t \left(Y^{(m)}_s\right)^{\alpha_2} dB_s
- b_2 \int_0^t L^{(m)}_s\, ds,
\quad t \ge 0.
\end{align*}
If we denote
$Z_t = Y^{(m)}_t - L^{(m)}_t$
and note that
$\ex \int_0^t (Y^{(m)}_s)^{\alpha_2} dB_s = 0$,
we get that
\[
\ex Z_t = Y_0 + a_2 t - b_2 \int_0^t \ex Z_s \,ds
- b_2 \int_0^t \ex L^{(m)}_s\, ds.
\]
Denote, for simplicity, 
$\ex Z_t = z_t$, $\ex L^{(m)}_t = \ell_t$.
Then we get an ODE
\[
z_t + b_2 \int_0^t z_s\,ds = Y_0 + a_2 t - b_2 \int_0^t \ell_t\, ds,
\]
or
\begin{equation}\label{eq:ode}
u_t' + b_2 u_t = Y_0 + a_2 t - b_2 \int_0^t \ell_s\,ds,
\end{equation}
where
$u_t = \int_0^t z_s\, ds$.
The solution of \eqref{eq:ode} has the form
\[
u_t = a_2 \int_0^t s e^{b_2(s-t)}\,ds - b_2 \int_0^t e^{b_2(s-t)} \int_0^s \ell_u\,du\,ds + \frac{Y_0}{b_2} \left(1 - e^{-b_2 t}\right),
\]
whence
\[
z_t = \frac{a_2}{b_2} \left(1 - e^{-b_2 t}\right)
- b_2 \int_0^t e^{b_2(s-t)} \ell_s \,ds + Y_0 e^{-b_2 t},
\]
or
\[
\ex Y^{(m)}_t = \ell_t + \frac{a_2}{b_2} \left(1 - e^{-b_2 t}\right)
- b_2 \int_0^t e^{b_2(s-t)} \ell_s \,ds + Y_0 e^{-b_2 t}.
\]
Let $t \to +\infty$.
Then $\ell_t$, being a nondecreasing function, tends to some limit $\ell_\infty$, and
\[
\lim_{t\to\infty} b_2 \int_0^t e^{b_2(s-t)} \ell_s \,ds
= \lim_{t\to\infty} \frac{b_2 e^{b_2 t} \ell_t}{b_2 e^{b_2t}} = \ell_\infty,
\]
and the proof follows.
\end{proof}
\begin{remark} Let us briefly discuss the level $m>0$ of reflection for the process  $Y$. In some sense,  any level is appropriate, because substituting $Y^{(m)}$ instead of $Y$ into $X$, according to comparison Theorem  \ref{comp}, and standard properties of $X$ with nonrandom drift, we get that in the worst case $\alpha_1=1/2$ $X$ can have only reflecting zeros from which it immediately attains strictly positive values. However, it still can spend some time under any positive level. And only if we choose $m\ge \frac12\sigma_1^2,$ then $X$ is strictly positive, again, due to comparison theorem. 
\end{remark}

\appendix

\section{External process in generalized form: existence, uniqueness and comparison theorem}

In this appendix, we consider the stochastic differential equation \eqref{eq:dX} for the \emph{external} process corresponding to a broader class of \emph{internal} processes, which are not necessarily solutions to \eqref{eq:dY}.
In subsection \ref{app:ex-un}, we prove the existence and uniqueness of a solution to this equation. In subsection \ref{comp}, we establish a version of the comparison theorem applicable to such equations.

\subsection{Existence and uniqueness}\label{app:ex-un}
\begin{theorem}\label{th:ex-un}
Consider equation of the form
\begin{equation}\label{eq_every}
X_t = X_0 + \int_0^t \left(U_s - b X_s\right) ds
+ \sigma  \int_0^t X_s^{\alpha}\, dW_s,
\end{equation}
where $X_0,   b, \sigma > 0$, $t \ge 0$,
and $U = \set{U_t, t \ge 0}$ is a continuous non-negative  process adapted to the filtration $\mathbb{F}.$
\begin{itemize}
\item [$(i)$]
Let $\alpha = 1$. Then the linear equation
\begin{equation}\label{eq_forU}
 X_t = X_0 + \int_0^t \left(U_s - b  X_s\right) ds
+ \sigma  \int_0^t X_s\, dW_s   
\end{equation}
has a unique strong solution of the form
\begin{equation}\label{linearrepr}
    X_t = \exp\set{-b  t - \frac{\sigma^2}{2} t + \sigma  W_t}
\left(X_0 + \int_0^t U_s \exp\set{b  s + \frac{\sigma^2}{2} s - \sigma  W_s} ds\right).
\end{equation}
This solution is strictly positive.
\item [$(ii)$]
Let $\alpha \in (\frac12, 1)$, the process $U$ satisfies the condition of item $(i)$, and for any $T>0$ $$\sup_{0\le t\le T}\ex U_t^2<\infty.$$ Then the   equation \eqref{eq_forU} has unique strong solution such that  for any $T>0$ \begin{equation}\label{sup_eq}
 \ex \sup_{0\le t\le T} X_t^2<\infty.   
\end{equation} 
\end{itemize} 
\end{theorem}
\begin{proof}
Item $(i)$ can be proved by direct calculations. In item $(ii)$, taking into account the existence of the second moment of $U$, the proof of   uniqueness is the same as in the standard Yamada--Watanabe theorem, see e.g., \cite[Theorem 3.2, page 182]{ikeda-watanabe}. As for existence, one can easily modify the proof of the existence theorem given in \cite[p.\ 59]{skor61}, again, taking into account the finiteness of $\sup_{0\le t\le T} \ex U_t^2$ for any $T>0$ and show that for the solution $X$ it holds that   $\sup_{0\le t\le T}\ex  X_t^2<\infty$.  Concerning \eqref{sup_eq}, note that for any $T > 0$,
\[
\sup_{0 \le t \le T} X_t^2 \le 3 X_0^2 + 3 \int_0^T U_s^2 \,ds
+ 3 \sup_{0 \le t \le T} \left(\int_0^t X_s^{\alpha_2} \,dW_s\right)^2.
\]
Then we can apply the Burkholder--Gundy inequality for the square-integrable martingale\linebreak $\int_0^t X_s^{\alpha_2} \,dW_s$ and get that
\begin{align*}
\ex \sup_{0 \le t \le T} X_t^2 &\le 3 X_0^2 + 3 \int_0^T \ex U_s^2 \,ds
+ 12 \int_0^T \ex X_s^{2\alpha_2} \,ds
\\
 &\le 3 X_0^2 + 3 \int_0^T \ex U_s^2 \,ds + 12 T
+ 12 \int_0^T \ex X_s^2 \,ds,
\end{align*}
and the proof immediately follows.
\end{proof}

\subsection{Comparison theorem}\label{comp}
Now let us establish the comparison theorem for the external processes.  This comparison result is absolutely obvious in the linear case when $\alpha_1=1,$ because it follows from representation \eqref{linearrepr} that if $U^i=\{U^i_t, t\ge 0\}$, $i=1,2$ are two continuous processes, adapted to the filtration $\mathbb{F}$, and with probability 1 it holds that $U^1_t\ge U_t^2$, $t\ge 0$, then respective solutions are in an analogous inequality. Therefore we now consider the equation \eqref{eq_every} in the case $\alpha_1\in[1/2,1)$.
\begin{theorem}\label{th:comparison}
Let $\set{U^i_t, t \ge 0}$, $i = 1, 2$, be  two continuous, non-negative processes satisfying the following conditions:
\begin{enumerate}[(i)]
\item for any $T > 0$ $\sup_{0 \le t \le T} \ex U^i_t < \infty$;
\item $U^1_0 = U^2_0$;
\item $U^1_t \ge U^2_t$, $t \ge 0$, with probability 1.
\end{enumerate}
Then the solutions of the equations
\[
X^i_t = X_0 + \int_0^t \left(U^i_s - b X^i_s\right) ds
+ \sigma  \int_0^t \left(X^i_s\right)^{\alpha}\, dW_s,
\]
satisfy the relation
\[
X^1_t \ge X^2_t, \; t \ge 0, \text{ with probability } 1.
\]
\end{theorem}
\begin{proof}
Introduce the process
\[
V_t = U^1_t - U^2_t - b \left(X^1_t - X^2_t\right), \quad t \ge 0,
\]
and consider the following stopping times:
\begin{gather*}
\tau_0 = 0, \quad
\tau_1 = \inf\set{t > 0 : V_t < 0}, \quad
\tau_2 = \inf\set{t > \tau_1 : V_t > 0}, \dots,\\
\tau_{2n+1} = \inf\set{t > \tau_{2n} : V_t < 0}, \quad
\tau_{2n+2} = \inf\set{t > \tau_{2n+1} : V_t > 0}, \quad n \ge 0.
\end{gather*}
Note that on any interval
$[\tau_{2n+1}, \tau_{2n+2}]$, $n \ge 0$,
we have that
\[
X^1_t - X^2_t \ge b^{-1} \left(U^1_t - U^2_t\right) \ge 0.
\]
Therefore, it is sufficient to consider only the intervals
$[\tau_{2k}, \tau_{2k+1}]$.
On any such interval we apply the standard method of Ikeda and Watanabe from \cite{ikewat}.
Namely, let
$\varphi_n(u)$, $n \ge 1$, $u \ge 0$,
be a non-negative continuous function such that its support is
$(a_n, a_{n-1})$, $\int_{a_n}^{a_{n-1}} \varphi_n(u) \,du = 1$,
$\varphi_n(u) \le \frac{2}{n u^{2\alpha}}$,
where the sequence
$a_0 = 1 > a_1 > \dots > a_n > \dots \downarrow 0$
is defined by
\[
\int_{a_n}^{a_{n-1}} u^{-2\alpha} \,du = \frac{1}{2\alpha - 1} \left(a_n^{1-2\alpha} - a_{n-1}^{1-2\alpha} \right) = n.
\]
Define also
\[
\psi_n(x) = \int_0^{\abs{x}} dy \int_0^y \varphi_n(z)\,dz,
\quad x \in \real, \quad n \ge 1.
\]
Then $\psi_n \in C^2(\real)$, $\psi_n(x) \uparrow \abs{x}$ as $n \to \infty$, and $\abs{\psi_n'(x)} \le 1$.
Using the It\^o formula, we can write
\begin{align*}
\MoveEqLeft
\psi_n \left(X^1_{t\wedge \tau_{2k+1}} - X^2_{t\wedge \tau_{2k+1}}\right)
- \left(X^1_{t\wedge \tau_{2k}} + X^2_{t\wedge \tau_{2k}}\right)
\\
&= \sigma \int_{t\wedge \tau_{2k}}^{t\wedge \tau_{2k+1}} \psi_n'\left(X^1_s - X^2_s\right) \left(\left(X^1_s\right)^\alpha - \left(X^2_s\right)^\alpha \right) dW_s 
\\
&\quad + \int_{t\wedge \tau_{2k}}^{t\wedge \tau_{2k+1}} \psi_n'\left(X^1_s - X^2_s\right) \left(U^1_s - U^2_s - b \left(X^1_s - X^2_s\right)\right) ds 
\\
&\quad + \frac12 \int_{t\wedge \tau_{2k}}^{t\wedge \tau_{2k+1}} \psi_n''\left(X^1_s - X^2_s\right) \left(\left(X^1_s\right)^\alpha - \left(X^2_s\right)^\alpha\right)^2 ds
\\
&\eqqcolon J_{1,n} + J_{2,n} + J_{3,n}.
\end{align*}
Then
\[
\ex J_{1,n} = 0,
\]
and 
\[
\ex J_{3,n} \le \ex \int_{t\wedge \tau_{2k}}^{t\wedge \tau_{2k+1}} \varphi_n\left(\abs{X^1_s - X^2_s}\right) \abs{X^1_s - X^2_s}^{2\alpha} ds
\le \frac{t}{n} \to 0 \quad \text{as } n \to \infty.
\]
Note also that
$\abs{\psi_n'(x)} \le 1$.
Therefore
\begin{align*}
\MoveEqLeft 
\ex\abs{X^1_{t \wedge \tau_{2k+1}} - X^2_{t \wedge \tau_{2k+1}}}
- \ex\abs{X^1_{t \wedge \tau_{2k}} - X^2_{t \wedge \tau_{2k}}}
\\
&\le \limsup_{n\to\infty} \ex J_{2,n}
\le \ex\int_{t\wedge \tau_{2k}}^{t\wedge \tau_{2k+1}} \abs{U^1_s - U^2_s - b \left(X^1_s - X^2_s\right)} ds 
\\
&=  \ex\int_{t\wedge \tau_{2k}}^{t\wedge \tau_{2k+1}} \Bigl(U^1_s - U^2_s - b \left(X^1_s - X^2_s\right)\Bigr) ds 
\\
&= \ex \left(X^1_{t \wedge \tau_{2k+1}} - X^2_{t \wedge \tau_{2k+1}} - X^1_{t \wedge \tau_{2k}} + X^2_{t \wedge \tau_{2k}}\right).
\end{align*}

Now let us use induction in $k$.
Obviously,
$X^1_{t \wedge \tau_0} - X^2_{t \wedge \tau_0} = 0$.
Assume that 
$X^1_{t \wedge \tau_{2k}} - X^2_{t \wedge \tau_{2k}} \ge 0$,
then
\[
\ex\abs{X^1_{t \wedge \tau_{2k+1}} - X^2_{t \wedge \tau_{2k+1}}}
\le \ex\left(X^1_{t \wedge \tau_{2k+1}} - X^2_{t \wedge \tau_{2k+1}}\right),
\]
whence
$X^1_{t \wedge \tau_{2k+1}} - X^2_{t \wedge \tau_{2k+1}} \ge 0$ a.s.
Since $\bigcup_{n=0}^\infty[\tau_{2n},\tau_{2n+1}] = \real^+$,
we get the proof.
\end{proof}

\providecommand{\MR}{\relax\ifhmode\unskip\space\fi MR }
\providecommand{\MRhref}[2]{%
  \href{http://www.ams.org/mathscinet-getitem?mr=#1}{#2}
}
\providecommand{\href}[2]{#2}

\end{document}